\documentclass[12 pt]{article}
\counterwithin*{equation}{section}
\usepackage[all]{xy}
\usepackage{amsmath,amssymb,amsfonts,latexsym,float,graphics,epsfig}

\usepackage{setspace}
\usepackage{mathtools}
\usepackage{xcolor}
\usepackage[top=2cm, bottom=2cm, left=2cm, right=2cm]{geometry}
\usepackage[colorlinks]{hyperref}
\usepackage{tikz}
\usetikzlibrary{cd}

\newtheorem{theorem}{Theorem}[section]
\newtheorem{proposition}[theorem]{Proposition}

\newtheorem{proof}[theorem]{Proof}

\begin{document}

\title{Discrete Dynamics on Locally Conformal Framework}  
\maketitle

\begin{center}
 \author{
 O\u{g}ul Esen\footnote{oesen@gtu.edu.tr}*, Ayten Gezici\footnote{agezici@gtu.edu.tr, corresponding author}*, Hasan Gümral\footnote{hgumral@yeditepe.edu.tr}**\\

  \bigskip
  
 *Department of Mathematics, Gebze Technical University,  \\  41400 Gebze-Kocaeli, Turkey  \\  
** Department of Mathematics, Yeditepe University,  
  \\   34755 Ataşehir-İstanbul, Turkey  }
\end{center}

\maketitle

\begin{abstract}
In this paper, we address the globalization problem of discrete Lagrangian and Hamiltonian dynamics in locally conformal framework.
\end{abstract}

\tableofcontents
\setlength{\parskip}{4mm}
\onehalfspacing

\section{Introduction}

 Discrete dynamics is gaining attention from both theoretical and practical perspectives due to its deep geometrical and algebraic foundations, as well as its effectiveness in solving equations numerically  \cite{BlanesCasas,Hairer}. Notably, discrete Hamiltonian dynamics \cite{Sanz-Serna,Yo,GeMarsden} and discrete Euler-Lagrange equations \cite{guibout2004discrete,MarsWest} are two significant approaches within the field of discrete geometry. The Hamilton-Jacobi theorem is also addressed in discrete models \cite{OhsawaBlochLeok, OhsawaBlochLeok2}.  In its most abstract form, Lie groupoids play a crucial role in discrete dynamics \cite{MarrDiegMart06}. In this framework, the coupling problem of two mutually interacting system is examined in \cite{esen2018matched}. Additionally, the study of implicit discrete dynamics and the Tulczyjew triplet for singular dynamics has been explored in detail in \cite{Leok-Tulc}. 
 
To clarify the purpose of this paper, let's briefly review some geometric constructions. 
Jacobi manifolds\footnote{Kirillov called this geometry as local Lie algebra \cite{Kirillov}.} represent generalizations of Poisson manifolds obtained by relaxing the Leibniz rule \cite{Lich,Lichnerowicz-Jacobi,Marle91,vaisman2002jacobi}. These manifolds have collected attention for their suitability in describing irreversible dynamic systems such as thermodynamics and dissipative models \cite{CarinenaGuha,SchaftMaschke}. Jacobi manifolds are completely integrable, where the odd-dimensional leaves of the Jacobi manifold are regarded as contact manifolds \cite{Arn,libermann2012symplectic}, while the even-dimensional leaves are considered as locally conformally symplectic (LCS) manifolds \cite{Banyaga,Vaisman85}. The discrete dynamics within contact geometry recently announced in the literature, see, for example, \cite{Bravetti2, Bravetti1, esen2022discrete,Vermeeren}. 

In the present work, we investigate discrete dynamics on locally conformally symplectic geometry. As we shall be more concrete in the main body of the paper, the study of discrete dynamics within the locally conformal symplectic framework addresses the globalization problem encountered in variational and symplectic integrators. Specifically, this paper explores the geometric aspects of local dynamical behaviour that fails to be extended globally.  

The paper is structured as follows: 
To arrive at a proper discretization, we begin by applying the inverse Legendre transformation to the continuous Hamiltonian dynamics on the LCS manifolds. This leads us to arrive at a novel framework called locally conformally Euler-Lagrange (LCEL) dynamics. The LCEL system represents a generalization of the classical Euler-Lagrange equations. Subsequently, we replace the action integral with a discrete sum and apply discrete variation. As a result, we introduce discrete locally conformally Euler-Lagrange (dLCEL) dynamics, which generalizes the discrete Euler-Lagrange equations (dEL). Finally, we employ the Legendre transformation to obtain the discrete locally conformally symplectic  (dLCS) Hamilton's equation within discrete geometry.

Accordingly, the paper is organized into two main sections. In the following section, following \cite{esen2023LCS-kin}, we present a quick summary of LCS geometry on the cotangent bundle. We further extend this locally conformal characterization to the tangent bundle and derive locally conformally Euler-Lagrange LCEL equations. In section \ref{discrete-sec}, we present the discretizations of LCS Hamiltonian dynamics and LCEL equations. These novel dynamics will be obtained on a purely geometrical basis, making them suitable for further numerical  as well as  theoretical analysis. 



\section{Locally Conformally Continuous Dynamics}

\subsection{Locally Conformally Symplectic Geometry}

A manifold $M$ is symplectic if it possesses a closed, non-degenerate two-form $\Omega$, called the symplectic two-form \cite{libermann2012symplectic,Marsden1999}. On a symplectic manifold, Hamilton's equation for a given Hamiltonian function $H$ is defined to be 
\begin{equation}
\iota_{X_H}\Omega = dH,
\end{equation}
where $X_H$ is the Hamiltonian vector field.

Consider the cotangent bundle $T^*Q$ of a manifold $Q$, where the canonical (Liouville) one-form $\Theta_{Q}$ exists. The canonical symplectic two-form on $T^*Q$ is given by the negative exterior derivative 
\begin{equation}
\Omega_Q=-d\Theta_Q
\end{equation}
of the canonical one-form $\Theta_{Q}$. This construction provides a foundational understanding of symplectic structures since locally every symplectic manifold is looking like a cotangent bundle. 

\bigskip
\noindent 
\textbf{Locally Conformally Symplectic Manifolds.} A manifold $M$ is called a locally conformally symplectic (LCS) manifold if it admits a non-degenerate two-form $\Omega$ so that 
\begin{equation}\label{eqw}
d\Omega=\varphi\wedge \Omega
\end{equation}
for a closed (called as Lee) one-form $\varphi$, see \cite{Bazzoni2018,LeeHC,Libermann-LCS,Vaisman85}.  We denote an LCS manifold by $(M,\Omega,\varphi)$. In terms of the Lichnerowicz-deRham differential\footnote{For a closed one-form $B$, the Lichnerowicz-deRham differential is defined to be $d_B\Omega=d\Omega-B\wedge \Omega$, see \cite{GuLi84}.} we can write the equality \eqref{eqw} as
\begin{equation}
d_\varphi\Omega=0.
\end{equation}
If the dimension of an LCS manifold is $2n$ then we determine the LCS volume-form as  
\begin{equation}\label{LCS-volume}
d \mu=\frac{(-1)^{n(n-1)/2}}{n!} \Omega\wedge \dots\wedge \Omega.
 \end{equation}
 
The local picture for an LCS manifold is as follows. On every local chart of the manifold $M$ there exists local symplectic two-form. We record a symplectic chart as $(U_\alpha,\Omega_\alpha)$. On the intersection $U_\alpha \cap U_\beta$ of two local charts, the local symplectic two-forms $\Omega_\alpha$ and $\Omega_\beta$ are related with 
\begin{equation}\label{sympl-rel-}
   e^{\sigma_\alpha}\Omega_\alpha = e^{\sigma_\beta}\Omega_\beta,
\end{equation}
where the family of real-valued functions $\{\sigma_\alpha \}$ are the conformal factors. 
It is noteworthy to see that the conformal factors satisfy the cocycle condition, allowing us to assemble local symplectic two-forms $\Omega_\alpha$ into a line bundle-valued two-form on $M$. We define a real-valued global two-form $\Omega$ using the following local characterization:
\begin{equation}
\Omega\vert_\alpha=e^{\sigma_\alpha}\Omega_\alpha,
\end{equation}
where $\Omega\vert_\alpha$ is the restriction of the global two-form to $U_\alpha$.
The exterior derivative of this two-form is computed to be 
\begin{equation}\label{hola}
 d\Omega\vert_\alpha=d\sigma_\alpha\wedge \Omega\vert_\alpha,
\end{equation}
which is the same as the LCS two-form condition in \eqref{eqw} with the Lee one-form $\varphi\vert_\alpha=d\sigma_\alpha$ defined as the exterior derivative of the conformal factor. 

\bigskip
\noindent 
\textbf{Cotangent Bundle and Coordinate Realization.} 
Generic examples of LCS manifolds are cotangent bundles  \cite{ChMu17}. 
 Consider a closed one-form $\vartheta$ on a manifold $Q$. We pull back this one-form to the cotangent bundle using the cotangent bundle projection and obtain a closed and semi-basic one-form 
\begin{equation}\label{varphi}
    \varphi = (\pi_Q)^* \vartheta.
\end{equation}
The negative LdR differential of the canonical one-form $\Theta_Q$ is a LCS two-form
\begin{equation}\label{Omega-phi}
\Omega_\varphi=-d_\varphi \Theta_Q = - d\Theta_Q + \varphi\wedge \Theta_Q  = \Omega_Q + \varphi\wedge \Theta_Q .
\end{equation} 
Therefore, the three-tuple  $
(T^*Q,\Omega_\varphi,\varphi)$  is an LCS manifold. 

From now on we intereted in the LCS manifold $
(T^*Q,\Omega_\varphi,\varphi)$. Let us be a bit more precise about this construction.  Assume an open covering for a manifold $Q$ with dimension $n$: 
\begin{equation}\label{chart}
Q=\bigsqcup_\alpha V_\alpha
\end{equation}
This induces manifold structures on the tangent and the cotangent bundle as
\begin{equation}\label{chartsTQ}
TQ=\bigsqcup_\alpha TV_\alpha = \bigsqcup_\alpha V_\alpha \times \mathbb{R}^n, \qquad T^*Q= \bigsqcup_\alpha T^*V_\alpha = \bigsqcup_\alpha V_\alpha \times (\mathbb{R}^n)^*,
\end{equation}
respectively. The Darboux coordinates $(q^i,p_i)$ on $T^*V_\alpha$ reads the global canonical one-form $\Theta_{Q}$ and the global symplectic two-form $\Omega_Q$ as
\begin{equation}\label{bar-bar}
\Theta_{Q}\vert_\alpha=p_idq^i, \qquad \Omega_Q\vert_\alpha=dq^i\wedge dp_i,
\end{equation}
respectively. On the other hand, there exist Darboux coordinates $(q^i,r_i)$ that yield the local canonical one-form $\Theta_\alpha$ (which is not global) and the local canonical symplectic two-form $\Omega_\alpha$ (which is also not global) as
\begin{equation}\label{r-p}
\Theta_\alpha = r_i dq^i, \qquad \Omega_\alpha=dq^i \wedge dr_i,
\end{equation}
respectively. The relationship between local momentum coordinates is 
\begin{equation}
r_i = e^{-\sigma_\alpha}p_i. 
\end{equation}
As presented in \eqref{Omega-phi}, the Lee one-form is a semi-basic one-form on the cotangent bundle framework. So, without a loss of generality,  we can assume that the conformal factors depend only on the coordinates $(q^i)$ of the base manifold $Q$. That is $\sigma_\alpha=\sigma_\alpha(q)$.

Accordingly, we can write the local symplectic two-form $\Omega_\alpha$ in \eqref{bar-bar} in terms of the coordinates $(q^i,p_i)$  as 
\begin{equation}
\Omega_\alpha= -d\Theta_\alpha = dq^i \wedge dr_i= e^{-\sigma_\alpha} \Big (dq^i\wedge dp_i + \frac{1}{2} A_{ij}dq^i\wedge dq^j \big), 
\end{equation}
where we have used the following notations: 
\begin{equation}\label{Aij}
A_{ij}:= \varphi_ip_j -\varphi_jp_i, \qquad \varphi_i:=\frac{\partial \sigma_\alpha}{\partial q^i}.
\end{equation}
Thus LCS two-form $\Omega_\varphi$ defined in \eqref{Omega-phi} is written in terms of the coordinates $(q^i,p_i)$ as
\begin{equation}\label{twist}
\Omega_\varphi\vert_\alpha = dq^i\wedge dp_i + \frac{1}{2} A_{ij}dq^i\wedge dq^j .
\end{equation}
See that the first term on the right-hand side is the canonical symplectic two-form $\Omega_Q$ whereas the second one is precisely $\varphi\wedge \Theta_Q$. Therefore, the local realization \eqref{twist} fits well with the global definition of $\Omega$ in \eqref{Omega-phi}. The following diagram summarizes the situation in one picture:
 \begin{equation*}
\xymatrix{
{\begin{array}{c} \Omega_Q\vert_\alpha \\ \eqref{bar-bar} \end{array}}   
 \ar[rrrd]^{\quad twist ~ \eqref{Omega-phi}}   & & & 
 {\begin{array}{c} \end{array}}  \\
  {\begin{array}{c} \Theta_Q\vert_\alpha \\ \eqref{bar-bar} \end{array}}   
 \ar[rrr]^{-d_\varphi} \ar[u]^{-d} &   & & 
 {\begin{array}{c} \Omega_\varphi\vert_\alpha \\
 \eqref{twist}\end{array}}  \\
{\begin{array}{c} \Theta_\alpha \\ \eqref{r-p}  \end{array}}  \ar[rrr]^{-d} \ar[u]^{e^{\sigma_\alpha}}&  & & {\begin{array}{c} \Omega_\alpha \\  \ar[u]^{e^{\sigma_\alpha}}\eqref{r-p} \end{array}} }
\end{equation*}
 We cite \cite{OtimanStanciu2017} for Darboux-Weinstein theorem for LCS manifolds. 

\bigskip
\noindent 
\textbf{LCS Hamiltonian Dynamics.} 
On a local symplectic space $(T^*V_\alpha,\Omega_\alpha)$ with conformal factors $(\sigma_\alpha)$, for a local Hamiltonian function $H_\alpha$, we write local Hamilton's equation as
\begin{equation}\label{geohamalpha}
    \iota_{X_{\alpha}}\Omega_{\alpha}=dH_{\alpha},
\end{equation}
where $X_{\alpha}$ is the local Hamiltonian vector field. Assuming the local equality 
\begin{equation}\label{Fatih}
e^{\sigma_\alpha}H_\alpha=e^{\sigma_\beta}H_\beta
\end{equation}
on the intersection $T^*V_\alpha\cap T^*V_\beta$, we define a global function $H$ on $T^*Q$ so that its restriction reads 
\begin{equation} \label{glueHamFunc}
H\vert_\alpha=e^{\sigma_\alpha}H_\alpha.
\end{equation}
In accordance with, we glue the local Hamilton's equation \eqref{geohamalpha} to the LCS Hamilton's equation
\begin{equation}\label{semiglobal}
   \iota_{\xi_{H}}\Omega=d_\varphi H.
\end{equation}
Here, the LCS Hamiltonian vector field $\xi_{H}$ is obtained by gluing together the local Hamiltonian vector fields $X_\alpha$ from \eqref{geohamalpha}, such that $
\xi_{H}\vert_\alpha=X_\alpha $. 
In terms of the Darboux coordinates $(q^i, p_i)$ of the canonical symplectic structure on the cotangent bundle $T^*Q$, the LCS Hamilton's equation \eqref{semiglobal} is expressed as:
\begin{equation}\label{LCSHE}
\frac{dq^i}{dt} = \frac{\partial H\vert_\alpha}{\partial p_i}, \qquad
\frac{dp_i}{dt} = -\frac{\partial H\vert_\alpha}{\partial q^i} - A_{ij} \frac{\partial H\vert_\alpha}{\partial p_j} + H\vert_\alpha\varphi_i,
\end{equation}
where the covariant quantity $A_{ij}$ is defined in \eqref{Aij}. We record also the divergence of the LCS Hamiltonian vector field:
\begin{equation} \label{div-xi-H}
\mathrm{div} (\xi_H)= \frac{1}{2}n \langle \varphi. \xi_H\rangle. 
\end{equation}
Notice that if there is no nontrivial conformal factor then LCS Hamilton's equation in \eqref{LCSHE} reduces to Hamilton's equation
\begin{equation}\label{HE}
   \frac{dq^i}{dt} = \frac{\partial H\vert_\alpha}{\partial p_i}, \qquad
\frac{dp_i}{dt} = -\frac{\partial H\vert_\alpha}{\partial q^i}.  
\end{equation}

\subsection{Locally Conformally Lagrangian Dynamics}\label{LCEL-sec} 

We follow the notations and definitions presented in the previous section. Consider  
a Lagrangian function $L_\alpha=L_\alpha(q^i,\dot{q}^i)$ define on the local tangent bundle $TV_\alpha$ and define the Legendre transformation
\begin{equation}\label{local-FL}
\mathbb{F}L_\alpha: TV_\alpha \longrightarrow T^*V_\alpha, \qquad (q^i,\dot{q}^i) \mapsto (q^i,r_i)=\big(q^i,\frac{\partial L_\alpha}{\partial \dot{q}^i}\big).
\end{equation}
We pull back the local canonical one-form and local symplectic two-form with this Legendre transformation and obtain local Poincar\'{e}-Cartan one-form and local Poincar\'{e}-Cartan two-form given by
\begin{equation}\label{loc-PC}
\begin{split}
        \theta_\alpha &= \mathbb{F}L^*_\alpha\Theta_\alpha = \frac{\partial L_\alpha}{\partial \dot{q}^i} dq^i 
        \\
        \omega_\alpha  &=  \mathbb{F}L^*_\alpha \Omega_\alpha =\frac{\partial^2 L_\alpha}{\partial \dot{q}^i \partial q^j} dq^i\wedge dq^j+\frac{\partial^2 L_\alpha}{\partial \dot{q}^i \partial \dot{q}^j} dq^i\wedge d\dot{q}^j,
\end{split}
\end{equation}
respectively. The local two-form $\omega_\alpha$ is symplectic if the rank of the Hessian matrix $[\partial^2 L_\alpha/\partial \dot{q}^i \partial \dot{q}^j ]$ is the maximum that is if the Lagrangian is regular. 

The local energy function $E_{L_\alpha}$ for a given local Lagrangian function $L_\alpha$ is defined by means of the Liouville vector field $\Delta$ as 
  \begin{equation}\label{locEnergy}
      E_{L_\alpha}= \Delta(L_\alpha)-L_\alpha, \qquad  \Delta=\dot{q}^i\frac{\partial}{\partial \dot{q}i}.
  \end{equation}
We introduce the following covariance relation
\begin{equation}\label{locEL}
    \iota_{Y_\alpha}\omega_\alpha=dE_{L_\alpha}
\end{equation}
on the tangent bundle $TV_\alpha$. This equation reads the local Euler-Lagrange equations 
\begin{equation}\label{EL-alpha}
\frac{d }{dt}  \frac{\partial L_\alpha}{\partial \dot{q}^i} =   \frac{\partial L_\alpha}{\partial q^i}.
\end{equation}

\bigskip
\noindent
\textbf{Locally conformally Poincar\'{e}-Cartan Forms.} 
Consider the tangent bundle $TQ$ covered by local charts $TV_\alpha$, each equipped with a local Lagrangian function. It is assumed that these local Lagrangian functions satisfy a conformal relation on intersections, denoted as $TV_\alpha \cap TV_\beta$, given by
\begin{equation}\label{laglog}
e^{\sigma_\alpha}L_\alpha =e^{\sigma_\beta}L_\beta,
\end{equation}
where ${\sigma_\alpha}$ represents a family of smooth functions. 
This construction implies that the local Lagrangian functions ${L_\alpha}$ cannot be seamlessly combined into a single real-valued Lagrangian function on the tangent bundle $TQ$. Instead, these local functions define a line-bundle-valued function on $TQ$. However, by multiplying the local functions by conformal factors, as indicated in \eqref{laglog}, a global real-valued function $L$ on $TQ$ emerges. The restriction of $L$ to an open chart $TV_\alpha$ is given by
\begin{equation}\label{loc-Lagrange}
L\vert_\alpha = e^{\sigma_\alpha}L_\alpha.
\end{equation}

This local-to-global realization is also true for local Poincar\'{e}-Cartan forms. Notice that local Poincar\'{e}-Cartan forms defined as \eqref{loc-PC} are related with conformal factors on the intersection $TV_\alpha \cap TV_\beta$. That is
\begin{equation}\label{dogu}
\begin{split}
         e^{\sigma_\alpha}\theta_\alpha &= e^{\sigma_\beta}\theta_\beta,
         \\
         e^{\sigma_\alpha}\omega_\alpha &= e^{\sigma_\beta}\omega_\beta,
\end{split}
\end{equation}
respectively. These local forms can only be glued up to a line bundle-valued one-form. Referring to \eqref{dogu}, we define global one-form $\theta_{L}$ and $\omega_\varphi$, in a respected manner, as follows
\begin{equation}\label{glo-glue-forms}
\begin{split}
\theta_{L}\vert_\alpha  &= e^{\sigma_\alpha}\theta_\alpha, \\
       \omega_\varphi \vert_\alpha &= e^{\sigma_\alpha}\omega_\alpha.
\end{split}
\end{equation}
Note that the negative exterior derivative of $\theta_L\vert_\alpha$ gives the following two-form
\begin{equation}\label{glo-PC-rest}
  \omega_L\vert_\alpha= - d\theta_L\vert_\alpha  =  
 -d\sigma_\alpha\wedge e^{\sigma_\alpha}\theta_\alpha + e^{\sigma_\alpha}\omega_\alpha.
\end{equation}
The global realizations of each term give that
\begin{equation}\label{conf-EP}
    \omega _\varphi = \omega_L + \varphi \wedge \theta_L,\qquad \omega _\varphi  = -d_\varphi\theta_L.
\end{equation}  
We call $\omega _\varphi$ as locally conformally Euler-Poincar\'{e} two-form. We present the following diagram to summarize the differential forms obtained in this subsection:
 \begin{equation}
\xymatrix{
{\begin{array}{c} \omega_L\vert_\alpha \\ \eqref{glo-PC-rest} \end{array}}   
 \ar[rrrd]^{\quad twist ~\eqref{conf-EP}}   & & & 
 {\begin{array}{c} \end{array}}  \\
  {\begin{array}{c} \theta_L\vert_\alpha \\ \eqref{glo-glue-forms} \end{array}}   
 \ar[rrr]^{-d_\varphi} \ar[u]^{-d} &   & & 
 {\begin{array}{c} \omega_\varphi\vert_\alpha \\
 \eqref{glo-glue-forms}\end{array}} \\
{\begin{array}{c} \theta_\alpha \\ \eqref{loc-PC}  \end{array}}  \ar[rrr]^{-d} \ar[u]^{e^{\sigma_\alpha}}&  & & {\begin{array}{c} \omega_\alpha \\  \ar[u]^{e^{\sigma_\alpha}}\eqref{loc-PC} \end{array}} }
\end{equation}

Assume that the energy function is also locally conformal that is    
\begin{equation}\label{globalenergyfnc}
E_{L}|_{\alpha}=e^{\sigma_{\alpha}} E_{L_\alpha}.
\end{equation}
By glueing the local Euler-Lagrange equation \eqref{locEL}, we arrive at the locally conformally Euler-Lagrange (LCEL) equation
\begin{equation}\label{LCEL-global}
     \iota_{\xi_L} \omega _\varphi =d_{\varphi}E_L
\end{equation}
  where the vector field $\xi_L$ is the global picture of the vector field $Y_\alpha$ that is $\xi_L\vert_\alpha=Y_\alpha$. 
      In the induced coordinates $(q^i,\dot{q}^i)$ on the tangent bundle $TQ$, the locally conformally Euler-Lagrange equations \eqref{LCEL-global} is written as \cite{esen2023LCS-kin}
         \begin{equation} \label{LCEL}
   \frac{d }{dt}(\frac{\partial L\vert _\alpha}{\partial \dot{q}^i}) - \frac{\partial L\vert _\alpha}{\partial q^i}= \varphi_j \dot{q}^j (\frac{\partial L\vert _\alpha}{\partial \dot{q}^i}) 
-\varphi_i L\vert _\alpha .
   \end{equation} 
  See that when the conformal factor is zero, we obtain the Euler-Lagrange equation:
\begin{equation}\label{EL}
\frac{d }{dt}(\frac{\partial L\vert _\alpha}{\partial \dot{q}^i}) - \frac{\partial L\vert _\alpha}{\partial q^i}=0
\end{equation}

\section{Locally Conformally Discrete Dynamics} \label{discrete-sec}

\subsection{Discrete Lagrangian and Discrete Hamiltonian Dynamics}

In continuous dynamics, Lagrangian function is defined on the tangent bundle $TQ$. In this case, solving the Euler-Lagrange equations \eqref{EL} involves determining a smooth curve $q = q(t)$ in $Q$. This curve, along with its velocity $(q(t), \dot{q}(t))$, is substituted into the Lagrangian function to achieve an extremum of the action integral. We cite \cite{MarsWest, OhsawaBlochLeok,OhsawaBlochLeok2} for more details on the discrete formulations presented in this subsection. 

To discretize the Lagrangian dynamics we replace the role of the smooth curve with a finite sequence of points. For example, let $a,b\in \mathbb{R}$ and $a<b$ then define the step $h=(b-a)/N$, where $N$ is the number of divisions, that is the number of elements in the finite sequence. We determine such a sequence by 
\begin{equation}
[\mathbf{q}]=(\mathbf{q}_a=\mathbf{q}_0,\dots,\mathbf{q}_b=\mathbf{q}_N).
\end{equation}
Here, the bold notation stands for $\mathbf{q}=(q^i)$. 
In the discrete framework, the Lagrangian function is substituted by a discrete Lagrangian function $L^d: Q\times Q\rightarrow \mathbb{R}$ which is an  approximation to the exact discrete Lagrangian 
\begin{equation}
L^{d-ex}(\mathbf{q}_\kappa,\mathbf{q}_{\kappa+1})=\int_{t_\kappa}^{t_{\kappa+1}}{L(\mathbf{q}(t),{\bf \dot{q}}(t))dt},
\end{equation}
where $\mathbf{q}:[t_\kappa,t_{\kappa+1}]\rightarrow Q$ is the solution of the continuous Euler-Lagrange equation with boundary conditions $\mathbf{q}(t_\kappa)=\mathbf{q}_\kappa$ and $\mathbf{q}(t_{\kappa+1})=\mathbf{q}_{\kappa+1}$. In the following proposition we state discrete Euler-Lagrange equations, see, for example, \cite{MarsWest}. 
\begin{proposition}\label{dEL-prop}
The discrete Euler-Lagrange (dEL) equation, arising from a discrete Lagrangian function $L^d$, is  
\begin{equation}\label{dEL}
D_2 L^d(\mathbf{q}_{\kappa-1},\mathbf{q}_\kappa) + D_1 L^d(\mathbf{q}_\kappa,\mathbf{q}_{\kappa+1}) = 0,
\end{equation}
where  $D_1$ denotes the partial derivative with respect to the first argument whereas $D_2$ is the partial derivative with respect to the second argument.
\end{proposition}
\begin{proof}
Let us write action sum for discrete Lagrangian $L^d$ as the following finite sum: 
 \begin{equation}
  S^d([\mathbf{q}])=   \sum_{\kappa=0}^{N-1}L^d(\mathbf{q}_\kappa, \mathbf{q}_{\kappa + 1}).
\end{equation}
We take the variation of this action and endow the discrete variational principle   $\delta S^d(\mathbf{q})=0 $. This gives us the equality:  
\begin{equation}
    \begin{split}
  \delta S^d(\mathbf{q})= \sum_{\kappa=0}^{N-1}(D_1L^d(\mathbf{q}_\kappa, \mathbf{q}_{\kappa + 1})\delta\mathbf{q}_\kappa + D_2L^d(\mathbf{q}_\kappa, \mathbf{q}_{\kappa + 1})\delta\mathbf{q}_{\kappa+1})=0.
    \end{split}
\end{equation}
Considering that at the boundary the variations are identically zero that is $\delta\mathbf{q}_0=\delta\mathbf{q}_N=0 $, we have the following expression: 
\begin{equation}
    \sum_{\kappa=1}^{N-1}D_1L^d(\mathbf{q}_\kappa, \mathbf{q}_{\kappa + 1})\delta\mathbf{q}_\kappa+  \sum_{\kappa=0}^{N-2}D_2L^d(\mathbf{q}_\kappa, \mathbf{q}_{\kappa + 1})\delta\mathbf{q}_{\kappa+1}=0.
\end{equation}
In the second term, we change $\kappa$ with $\kappa-1$ in order to take the whole expression into the parenthesis of $\delta\mathbf{q}_{\kappa}$. This reads 
\begin{equation}
 \sum_{\kappa=1}^{N-1}(D_1L^d(\mathbf{q}_\kappa, \mathbf{q}_{\kappa + 1}) + D_2L^d(\mathbf{q}_{\kappa-1}, \mathbf{q}_{\kappa }))\delta\mathbf{q}_{\kappa}=0.  
\end{equation}
For each point in the sequence, we get the discrete Euler-Lagrange equation \eqref{dEL}.  
\end{proof}

\bigskip
\noindent
\textbf{The Discrete Legendre Transformation.} As in the case of continuous dynamics, we employ the Legendre transformation to pass from the discrete Lagrangian dynamics to the discrete Hamiltonian dynamics. In the discrete case, we have the right and left Legendre transformations 
\begin{equation}\label{dFL}
\mathbb{F}^{\pm}L^{d}:Q\times Q \longrightarrow T^*Q
\end{equation}
given respectively by 
\begin{equation}\label{legtrans}
\begin{split}
 \mathbb{F}^+L^{d}(\mathbf{q}_\kappa,\mathbf{q}_{\kappa+1})&=(\mathbf{q}_{\kappa+1},D_2L^d(\mathbf{q}_\kappa,\mathbf{q}_{\kappa+1})), \\
\mathbb{F}^-L^{d}(\mathbf{q}_\kappa,\mathbf{q}_{\kappa+1})&=(\mathbf{q}_{\kappa}, -D_1L^d(\mathbf{q}_\kappa,\mathbf{q}_{\kappa+1})).
\end{split}
\end{equation}
We determine the momenta as the second entry of the image space of the discrete Legendre transformations. Accordingly, we have the right and the left momenta as 
\begin{equation}
    \begin{split}
        &\mathbf{p}^+_{\kappa,\kappa +1}=D_2L^d(\mathbf{q}_\kappa,\mathbf{q}_{\kappa+1}),\\
        &\mathbf{p}^-_{\kappa,\kappa +1} = -D_1L^d(\mathbf{q}_\kappa,\mathbf{q}_{\kappa+1}),
    \end{split}
\end{equation}
respectively. 
See that from the discrete Euler-Lagrange equation \eqref{dEL} we can define momenta depending on a single point as 
\begin{equation}
\mathbf{p}_\kappa:=\mathbf{p}^+_{\kappa-1,\kappa }=\mathbf{p}^-_{\kappa,\kappa +1}.
\end{equation}
Accordingly, we can write
\begin{equation}\label{glo-momenta}
    \begin{split}
   &\mathbf{p}_\kappa=  -D_1L^d(\mathbf{q}_\kappa,\mathbf{q}_{\kappa+1}),\\  &\mathbf{p}_{\kappa+1}=D_2L^d(\mathbf{q}_\kappa,\mathbf{q}_{\kappa+1}).
    \end{split}
\end{equation}
In the continuous case, the Legendre transformation is non-degenerate if the Lagrangian function is regular, meaning that the rank of the Hessian of the Lagrangian function is at its maximum. To analyze this phenomenon in the discrete case, we introduce discrete (right and left) Poincar\'{e}-Cartan one-forms:
\begin{equation}\label{PC-1}
\begin{split}
&\theta^+=D_2L^d(\mathbf{q}_{\kappa},\mathbf{q}_{\kappa+1})d\mathbf{q}_{\kappa+1},\\
&\theta^-=-D_1L^d(\mathbf{q}_{\kappa},\mathbf{q}_{\kappa+1})d\mathbf{q}_{\kappa}.
\end{split}
\end{equation}
The exterior derivatives of these one-forms coincide and determine the discrete Poincar\'{e}-Cartan two-form as
\begin{equation}\label{dPC}
 \omega^d(\mathbf{q}_\kappa,\mathbf{q}_{\kappa+1})=-d\theta^{\pm}=-D_1D_2L^d(\mathbf{q}_\kappa,\mathbf{q}_{\kappa+1})\ d\mathbf{q}_{\kappa}\wedge d\mathbf{q}_{\kappa+1}.
\end{equation}
The discrete Legendre transformation \eqref{dFL} is a local diffeomorphism if the discrete Poincar\'{e}-Cartan two-form is non-degenerate (hence symplectic). In this case, the discrete Lagrangian function is called regular. We determine the discrete Hamiltonian dynamics under this assumption.   

For a regular discrete Lagrangian function $L^d$, referring to the Darboux coordinates, we define the right discrete Hamiltonian function as the difference between the  coupling function and the discrete Lagrangian function as
\begin{equation}\label{rd-Ham}
     H^{d+}(\mathbf{q}_\kappa, \mathbf{p}_{\kappa+1})=  \mathbf{p}_{\kappa+1} \cdot  \mathbf{q}_{\kappa+1}- L^d(\mathbf{q}_\kappa, \mathbf{q}_{\kappa + 1}).
\end{equation} 
On the other hand, we define the left discrete Hamiltonian function as the difference between the \textit{minus} of the coupling function and the discrete Lagrangian function as
\begin{equation}\label{ld-Ham}
     H^{d-}(\mathbf{q}_{\kappa+1}, \mathbf{p}_{\kappa})= - \mathbf{p}_{\kappa} \cdot  \mathbf{q}_{\kappa}- L^d(\mathbf{q}_\kappa, \mathbf{q}_{\kappa + 1}).
\end{equation}
The following proposition presents the discrete Hamiltonian dynamics for both the left and right Hamiltonian functions. 
\begin{proposition} \label{dHam-main}
    The right discrete Hamilton's equations generated by the right Hamiltonian function \eqref{rd-Ham} are
    \begin{equation}
\begin{cases}\label{rd-Hameq}
&\mathbf{q}_{\kappa+1}= D_2H^{d+}(\mathbf{q}_\kappa, \mathbf{p}_{\kappa+1}), \\
&  \mathbf{p}_{\kappa}=D_1H^{d+}(\mathbf{q}_\kappa, \mathbf{p}_{\kappa+1}),
\end{cases}
\end{equation}
whereas the left discrete Hamilton's equations generated by the left Hamiltonian function \eqref{ld-Ham} are
\begin{equation}\label{ld-Hameq}
    \begin{cases}
        &\mathbf{q}_\kappa=-D_2H^{d-}(\mathbf{q}_{\kappa+1}, \mathbf{p}_{\kappa}), \\
        & \mathbf{p}_{\kappa + 1}=-D_1 H^{d-}(\mathbf{q}_{\kappa+1}, \mathbf{p}_{\kappa}),
    \end{cases}
\end{equation}
where $D_a$ denotes the partial derivative with respect to the $a$-th entry.
\end{proposition}
\begin{proof}
We consider the right Hamiltonian function \eqref{rd-Ham}, then take the partial derivatives of it with respect to its first entry $\mathbf{q}_{\kappa}$ and its second entry $\mathbf{p}_{\kappa+1}$. This gives us
 \begin{equation}
 \begin{split}
     &D_1H^{d+}(\mathbf{q}_\kappa, \mathbf{p}_{\kappa+1})=-D_1L^d(\mathbf{q}_\kappa, \mathbf{q}_{\kappa + 1})\\
     & D_2H^{d+}(\mathbf{q}_\kappa, \mathbf{p}_{\kappa+1})= \mathbf{q}_{\kappa+1}.
      \end{split}
 \end{equation}
Substituting the momentum definition from \eqref{glo-momenta}, we obtain the right discrete Hamilton's equations \eqref{rd-Hameq}.  

On the other hand, in the left case, we compute the partial derivatives of the left Hamiltonian function \eqref{ld-Ham} with respect to its first entry $\mathbf{q}_{\kappa+1}$ and its second entry $\mathbf{p}{\kappa}$, yielding:
\begin{equation}
\begin{split}
   & D_1 H^{d-}(\mathbf{q}_{\kappa+1}, \mathbf{p}_{\kappa})= -D_2L^d(\mathbf{q}_\kappa, \mathbf{q}_{\kappa + 1}) \\
   & D_2H^{d-}(\mathbf{q}_{\kappa+1}, \mathbf{p}_{\kappa})=-\mathbf{q}_\kappa
    \end{split}
\end{equation}
Substituting \eqref{glo-momenta}, we derive the left discrete Hamilton's equations \eqref{ld-Hameq}.  
\end{proof}

\subsection{Locally Conformally Discrete Lagrangian Dynamics}

In this subsection, we generalize the discrete dynamics discussions presented in the previous subsection to the locally conformal framework. In this novel construction, we start with elaborating the discretization of the locally conformally Euler-Lagrange (LCEL) dynamics given in Subsection \ref{LCEL-sec}. 

For a manifold $Q$, we assume an open covering as in \eqref{chart}, consisting of local charts $\{V_\alpha\}$. For each local chart $V_\alpha$, we define the discrete local Lagrangian function $L^d_\alpha$ on the product vector space $V_\alpha \times V_\alpha$ as follows
\begin{equation}
L^d_\alpha: V_\alpha \times V_\alpha \longrightarrow \mathbb{R}. 
\end{equation}
Referring to Proposition \ref{dEL-prop} and the dynamical equation \eqref{dEL}, we express the local discrete dynamics generated by this local discrete Lagrangian as
\begin{equation}\label{dEL-loc}
D_2 L^d_\alpha(\mathbf{q}_{\kappa-1},\mathbf{q}_\kappa) + D_1 L^d_\alpha(\mathbf{q}_\kappa,\mathbf{q}_{\kappa+1}) = 0.
\end{equation}

Recalling from the construction in Subsection \ref{LCEL-sec}, we assume locally conformal character  determined by the family of scalar functions $ \{\sigma_{\alpha} \}$ so that the local discrete Lagrangian functions are satisfying 
\begin{equation}\label{laglogd}
e^{\sigma_\alpha}L^d_\alpha =e^{\sigma_\beta}L^d_\beta
\end{equation}
on the intersection $V_\alpha \cap V_\beta $. 
Following \eqref{loc-Lagrange} in Subsection \ref{LCEL-sec}, we define a global discrete Lagrangian function $L^d$ on the whole manifold $Q$ whose local realization is 
\begin{equation}\label{dgLagfunc}
L^d \vert_\alpha = e^{\sigma_\alpha}L^d_\alpha . 
\end{equation}
The following proposition establishes discrete dynamics in terms of the global discrete Lagrangian function $L^d$, presenting the discretization of the locally conformal Euler-Lagrange equations \eqref{LCEL}.
\begin{proposition} \label{dLCEL-prop}
The discrete dynamics generated by the global discrete Lagrangian function $L^d$ defined in \eqref{dgLagfunc} is given by:
\begin{equation}\label{dLCEL}
\begin{split}
   & e^{\sigma_{\alpha}(\mathbf{q}_{\kappa})-\sigma_{\alpha}(\mathbf{q}_{\kappa-1})}D_2 L^d|{\alpha}(\mathbf{q}_{\kappa-1},\mathbf{q}_\kappa) \\ &  \qquad \qquad =D \sigma_{\alpha}(\mathbf{q}_{\kappa}) \cdot L^d|{\alpha}(\mathbf{q}_\kappa,\mathbf{q}_{\kappa+1})- D_1 L^d|_{\alpha}(\mathbf{q} _{\kappa},\mathbf{q}_{\kappa+1}),
\end{split}
\end{equation}
where $D_a$ denotes the partial derivative with respect to the $a$-th entry whereas $D \sigma_{\alpha}$ is the partial derivative of the function $\sigma_{\alpha}$. 
\end{proposition}
 \begin{proof}
We begin with the discrete action 
\begin{equation}\label{serkan}
  S^d([\mathbf{q}])=   \sum_{\kappa=0}^{N-1}L_\alpha^d(\mathbf{q}_\kappa, \mathbf{q}_{\kappa + 1})
\end{equation}
where $L^d_\alpha$ represents the local discrete Lagrangian function. To write the action sum for the global discrete Lagrangian function, we substitute \eqref{dgLagfunc} into the discrete action \eqref{serkan} yields
\begin{equation}
  S^d([\mathbf{q}])=   \sum_{\kappa=0}^{N-1}e^{-\sigma_{\alpha}(\mathbf{q}_\kappa)}L^d|_{\alpha}(\mathbf{q}_\kappa, \mathbf{q}_{\kappa + 1}) . 
\end{equation}
See that the action sum takes different values for different sequences. The variation of the action is computed to be
\begin{equation}
    \begin{split}
         \delta S^d([\mathbf{q}])&= \sum_{\kappa=0}^{N-1}(-e^{-\sigma_\alpha(\mathbf{q}_\kappa)}D\sigma_\alpha(\mathbf{q}_\kappa)L^d|_{\alpha}(\mathbf{q}_\kappa, \mathbf{q}_{\kappa + 1})\delta\mathbf{q}_\kappa \\
         &\qquad \qquad+e^{-\sigma_\alpha(\mathbf{q}_\kappa)}D_1L^d|_{\alpha}(\mathbf{q}_\kappa, \mathbf{q}_{\kappa + 1})\delta\mathbf{q}_\kappa \\
         & \qquad \qquad+ e^{-\sigma_\alpha(\mathbf{q}_\kappa)}D_2L^d|_{\alpha}(\mathbf{q}_\kappa, \mathbf{q}_{\kappa + 1})\delta\mathbf{q}_{\kappa+1}).
    \end{split}
\end{equation}
In order to extremize the action sum, we set $\delta S^d([\mathbf{q}])=0$ while fixing the boundary conditions: $\delta \mathbf{q}_{0}=0$ and $\delta \mathbf{q}_{N}=0$. Accordingly, we obtain
\begin{equation}
    \begin{split}
       &\sum_{\kappa=1}^{N-1}e^{-\sigma_\alpha(\mathbf{q}_\kappa)}(-D\sigma_\alpha(\mathbf{q}_\kappa)L^d|_{\alpha}(\mathbf{q}_\kappa, \mathbf{q}_{\kappa + 1})+D_1L^d|_{\alpha}(\mathbf{q}_\kappa, \mathbf{q}_{\kappa + 1}))\delta\mathbf{q}_\kappa\\
     & \qquad \qquad +\sum_{\kappa=0}^{N-2}e^{-\sigma_\alpha(\mathbf{q}_\kappa)}D_2L^d|_{\alpha}(\mathbf{q}_\kappa, \mathbf{q}_{\kappa + 1})\delta\mathbf{q}_{\kappa+1}=0 .
    \end{split}
\end{equation}
By relabeling $\kappa$ with $\kappa-1$ in the last term we write  
\begin{equation}
    \begin{split}
       &\sum_{\kappa=1}^{N-1}\Big(e^{-\sigma_\alpha(\mathbf{q}_\kappa)}(-D\sigma_\alpha(\mathbf{q}_\kappa)L^d|_{\alpha}(\mathbf{q}_\kappa, \mathbf{q}_{\kappa + 1})+D_1L^d|_{\alpha}(\mathbf{q}_\kappa, \mathbf{q}_{\kappa + 1}))\\
       & \qquad \qquad +e^{-\sigma_\alpha(\mathbf{q}_{\kappa-1})}D_2L^d|_{\alpha}(\mathbf{q}_{\kappa-1}, \mathbf{q}_{\kappa })\Big)\delta\mathbf{q}_{\kappa}=0.
    \end{split}
\end{equation}
Thus for arbitrary $\delta\mathbf{q}_{\kappa}$, the equality holds if \eqref{dLCEL} is satisfied. 
\end{proof}

We call \eqref{dLCEL} as discrete locally conformally Euler-Lagrange (dLCEL) equations. It is noteworthy that if there is no locally conformal character within the setting, meaning $L\vert_\alpha = L_\alpha$, then Proposition \ref{dLCEL-prop} aligns with Proposition \ref{dEL-prop}, and the discrete locally conformally Euler-Lagrange equation \eqref{dLCEL} reduces to the discrete Euler-Lagrange equation \eqref{dEL}. The following diagram summarizes this discussion.
\begin{equation}\label{d-d}
\xymatrix{{\begin{array}{c} LCEL\\ \eqref{LCEL} \end{array}}   
\ar[d]_{discretization} \ar[rrr]^{\varphi=0}&   & &  {\begin{array}{c} ~EL\\
\eqref{EL}  \end{array}} \ar[d]_{discretization} \\
{\begin{array}{c} dLCEL \\ \eqref{dLCEL}  \end{array}}  \ar[rrr]^{\varphi=0}&  & & {\begin{array}{c} dEL \\ \eqref{dEL} \end{array}} }
\end{equation}

\bigskip
\noindent 
\textbf{The Discrete Legendre Transformation.} For a local chart $V_\alpha$ we introduce the local right and left Legendre Transformations as 
\begin{equation}\label{dFL-loc}
\mathbb{F}^{\pm}L^d_\alpha: V_\alpha \times V_\alpha \longrightarrow T^*V_\alpha
\end{equation}
given explicitly by 
\begin{equation}\label{legtrans-loc}
\begin{split}
\mathbb{F}^+L^{d}_\alpha(\mathbf{q}_\kappa,\mathbf{q}_{\kappa+1})&=(\mathbf{q}_{\kappa+1},r^+_{\kappa,\kappa +1}=D_2L^d_\alpha(\mathbf{q}_\kappa,\mathbf{q}_{\kappa+1})), \\
\mathbb{F}^-L^{d}_\alpha(\mathbf{q}_\kappa,\mathbf{q}_{\kappa+1})&=(\mathbf{q}_{\kappa},r^-_{\kappa,\kappa +1} = -D_1L^d_\alpha(\mathbf{q}_\kappa,\mathbf{q}_{\kappa+1})),
\end{split}
\end{equation}
respectively. 
Here, $r^+_{\kappa,\kappa +1}$ represents the right local momenta, while $r^-_{\kappa,\kappa +1}$ denotes the left local momenta. 
To express the discrete Legendre transformations in terms of the global discrete Lagrangian function $L^d$, we utilize the identification \eqref{dgLagfunc} and substitute it into the definitions of local momenta. Consequently, we obtain 
\begin{equation}\label{r+-}
\begin{split}
\mathbf{r}^+_{\kappa,\kappa +1} & =e^{-\sigma_{\alpha}(\mathbf{q}_\kappa)} D_2 L^d|_{\alpha}(\mathbf{q}_\kappa, \mathbf{q}_{\kappa + 1}),
\\
 \mathbf{r}^-_{\kappa,\kappa +1} & = e^{-\sigma_{\alpha}(\mathbf{q}_\kappa)} \big( D\sigma_{\alpha} (\mathbf{q}_\kappa)\cdot L^d|_{\alpha}(\mathbf{q}_\kappa, \mathbf{q}_{\kappa + 1}) - D_1L^d|_{\alpha}(\mathbf{q}_\kappa, \mathbf{q}_{\kappa + 1})\big).
 \end{split}
\end{equation}
Under dLCEL equation \eqref{dLCEL}, and after a parameter shift, the right and the left momenta turn out to be the same and this results in the following definition 
\begin{equation}\label{r-momenta}
\mathbf{r}_\kappa:=\mathbf{r}^+_{\kappa-1,\kappa }=\mathbf{r}^-_{\kappa,\kappa +1}.
\end{equation}
Recall that in the continuous case, there exist two sets of momenta in a relation given by $r_i = e^{-\sigma_\alpha}p_i$. This is depicted in \eqref{r-p}. Motivated by this, and following the order in  \eqref{r+-}, we define two new sets of local momenta
\begin{equation}
\begin{split}
\mathbf{p}^+_{\kappa,\kappa +1} &=e^{\sigma_{\alpha}(\mathbf{q}_\kappa)}r^+_{\kappa,\kappa +1}= D_2 L^d|_{\alpha}(\mathbf{q}_\kappa, \mathbf{q}_{\kappa + 1}).
\\
\mathbf{p}^-_{\kappa,\kappa +1}&= e^{\sigma_{\alpha}(\mathbf{q}_\kappa)} r^-_{\kappa,\kappa +1} = D \sigma_{\alpha} (\mathbf{q}_\kappa)\cdot L^d|_{\alpha}(\mathbf{q}_\kappa, \mathbf{q}_{\kappa + 1}) - D_1L^d|_{\alpha}(\mathbf{q}_\kappa, \mathbf{q}_{\kappa + 1}).
\end{split}
\end{equation}
A naive shift of parameters is not enough to define a momenta $\mathbf{p}$ depending on a single point. More precisely one has that 
\begin{equation} \label{p-momenta}
\mathbf{p}_\kappa:=e^{\sigma_{\alpha}(\mathbf{q}_\kappa)-\sigma_{\alpha}(\mathbf{q}_{\kappa-1})} \mathbf{p}^+_{\kappa-1,\kappa }= \mathbf{p}^-_{\kappa,\kappa +1}.
\end{equation}
It is immediate to see that $\mathbf{r}$-momenta in \eqref{r-momenta} and $\mathbf{p}$-momenta in \eqref{p-momenta} are related as 
\begin{equation}\label{r-p-dd}
\mathbf{r}_{\kappa}= e^{-\sigma_{\alpha}(\mathbf{q}_{\kappa})}\mathbf{p}_{\kappa}
\end{equation}
as we expected. 

\bigskip
\noindent 
\textbf{Poincar\'{e}-Cartan Forms.} Referring to \eqref{PC-1}, we pull back the local canonical one-form $\Theta_\alpha$ in \eqref{r-p}, defined on the cotangent bundle $T^*V_\alpha$, to the product space $V_\alpha \times V_\alpha$ via the right and left Legendre transformations \eqref{legtrans-loc}. Then, in a respective manner, we define the right and left local Poincar\'{e}-Cartan one-forms for the local Lagrangian function $L_\alpha$ as follows
 \begin{equation}
    \begin{split}
        \theta_\alpha^+ &= \mathbb{F}^+L^{d}_\alpha(\Theta_\alpha) = D_2L_\alpha^d(\mathbf{q}_\kappa, \mathbf{q}_{\kappa + 1})d\mathbf{q}_{\kappa + 1},\\ 
  \theta^{-}_\alpha &=\mathbb{F}^-L^{d}_\alpha(\Theta_\alpha) =-D_1L^d_\alpha(\mathbf{q}_\kappa,\mathbf{q}_{\kappa+1})d\mathbf{q}_{\kappa}.
    \end{split}
\end{equation}
 Exterior derivatives of these one-forms coincide and determine local Poincar\'{e}-Cartan two-form
 \begin{equation}
\omega^+_{\alpha}=-d\theta^\pm_{\alpha}
\end{equation}
This is a symplectic two-form if the local Lagrangian function is regular.

Let us use the discrete Legendre transformations to elucidate the locally conformal character of Proposition \ref{dLCEL-prop} and the dLCEL equation presented in \eqref{dLCEL}. We do this for the right Poincar\'{e}-Cartan one-form the left case is more or less following the same lines of calculation. Referring to \eqref{dgLagfunc}, we write the right Poincar\'{e}-Cartan one-form in terms of the global discrete Lagrangian function as
\begin{equation}
\begin{split}
\theta^+_{\alpha} &=D_2(e^{-\sigma_{\alpha}(\mathbf{q}_{\kappa})}L^d|_{\alpha}(\mathbf{q}_\kappa, \mathbf{q}_{\kappa + 1})) d\mathbf{q}_{\kappa + 1} \\
    & = e^{-\sigma_{\alpha}(\mathbf{q}_{\kappa})}D_2L^d|_{\alpha}(\mathbf{q}_\kappa, \mathbf{q}_{\kappa + 1}) d\mathbf{q}_{\kappa + 1}.
    \end{split}
\end{equation}
This one-form fails to be global as in the continuous case. To have a global one we introduce a one-form $\theta^{d+}$ whose restriction to $V_\alpha$ is 
\begin{equation}
    \theta^{d+}|_{\alpha}= e^{\sigma_{\alpha}}\theta^+_{\alpha} =D_2L^d|_{\alpha}(\mathbf{q}_\kappa, \mathbf{q}_{\kappa + 1}) d\mathbf{q}_{\kappa + 1}.
\end{equation}
Taking the negative exterior derivative of this one-form we obtain a global two-form 
\begin{equation}\label{Serkan2}
\begin{split}
    \omega^{d+}|_{\alpha}&=- d \theta^{d+}|_{\alpha}= - e^{\sigma_{\alpha}(\mathbf{q}_{\kappa})}d\theta^+_{\alpha} -e^{\sigma_{\alpha}(\mathbf{q}_{\kappa})} d\sigma_{\alpha}(\mathbf{q}_{\kappa}) \wedge \theta^+_{\alpha} 
    \\ 
&=e^{\sigma_{\alpha}(\mathbf{q}_{\kappa})} \big(\omega^+_{\alpha}-d\sigma_{\alpha}(\mathbf{q}_{\kappa})\wedge \theta^+_{\alpha}\big ). 
\end{split}
\end{equation}
We  introduce discrete locally conformally Poincar\'{e}-Cartan two-form as
\begin{equation}
 \omega_{\varphi}^{d+}\vert_{\alpha}=e^{\sigma_{\alpha}(\mathbf{q}_{\kappa})} \omega^+_{\alpha},
\end{equation}
where $ \omega_{\varphi}^{d+}\vert_{\alpha} $ denotes the restriction of $ \omega_{\varphi}^{d+}$ to local chart $TV_{\alpha}$. 
Referring to \eqref{Serkan2}, it is a direct observation to see that  
\begin{equation}\label{glo-Lag}
     \omega_{\varphi}^{d+}\vert_{\alpha} = \omega^{d+}|_{\alpha}+ d\sigma_{\alpha}(\mathbf{q}_{\kappa})\wedge\theta^+|_{\alpha}.
\end{equation}
is satisfying the LCS two-form construction in \eqref{Omega-phi}. 
\begin{proposition}
    The discrete locally conformally Poincar\'{e}-Cartan two-form $\omega_{\varphi}^{d+}$ defined in \eqref{glo-Lag} is a LCS two-form with the Lee one-form is $d\sigma_{\alpha}$. 
\end{proposition}
\begin{proof}
Let us now perform the following calculation to assure once more that the discrete locally conformally Poincar\'{e}-Cartan  two-form $ \omega_{\varphi}^{d+}$ is indeed an LCS two-form:
\begin{equation}
\begin{split}
    \omega_{\varphi}^{d+}\vert_{\alpha} & =-d(D_2L^d|_{\alpha}(\mathbf{q}_\kappa, \mathbf{q}_{\kappa + 1}) d\mathbf{q}_{\kappa + 1} ) \\
    &\qquad \qquad + d\sigma_{\alpha}(\mathbf{q}_\kappa)\wedge D_2L^d|_{\alpha}(\mathbf{q}_\kappa,\mathbf{q}_{\kappa + 1} )d\mathbf{q}_{\kappa + 1} \\
    &=   -\frac{\partial^2 L^d|_{\alpha}}{\partial  \mathbf{q}^i_{\kappa }\partial \mathbf{q}^j_{\kappa+1 } }d \mathbf{q}^i_{\kappa }\wedge d \mathbf{q}^j_{\kappa+1 } -\frac{\partial^2 L^d|_{\alpha}}{\partial  \mathbf{q}^i_{\kappa+1 }\partial \mathbf{q}^j_{\kappa+1 } }d \mathbf{q}^i_{\kappa+1 }\wedge d \mathbf{q}^j_{\kappa+1 }\\
    &\hspace{2cm}  + \frac{\partial \sigma_{\alpha}}{\partial \mathbf{q}^i_{\kappa } } \frac{\partial L^d|_{\alpha}}{\partial \mathbf{q}^j_{\kappa+1 } }d \mathbf{q}^i_{\kappa }\wedge d \mathbf{q}^j_{\kappa+1 }. 
    \end{split}
\end{equation}
Then we take the exterior derivative of this explicit expression as 
\begin{equation}
\begin{split}
d \omega_{\varphi}^{d+}\vert_{\alpha} &=  \frac{\partial \sigma_{\alpha}}{\partial \mathbf{q}^i_{\kappa } }\frac{\partial^2 L^d|_{\alpha}}{\partial \mathbf{q}^l_{\kappa }\partial \mathbf{q}^j_{\kappa+1 } }d \mathbf{q}^l_{\kappa }\wedge d \mathbf{q}^i_{\kappa } \wedge d \mathbf{q}^j_{\kappa+1 }  \\
& = \frac{\partial \sigma_{\alpha}}{\partial \mathbf{q}^i_{\kappa } }d \mathbf{q}^i_{\kappa } \wedge (-\frac{\partial^2 L^d|_{\alpha}}{\partial \mathbf{q}^l_{\kappa }\partial \mathbf{q}^j_{\kappa+1 } }d \mathbf{q}^l_{\kappa }\wedge d \mathbf{q}^j_{\kappa+1 }  )\\
&= d\sigma_{\alpha}(\mathbf{q}_{\kappa}) \wedge  \omega^+\vert_{\alpha}. 
    \end{split}
    \end{equation}
reading that  that the discrete locally conformally Poincar\'{e}-Cartan  two-form satisfies the LCS two-form condition in \eqref{hola}. 
\end{proof}

\subsection{Locally Conformally Discrete Hamiltonian Dynamics} 

We have defined the discrete locally conformal Euler-Lagrange equations \eqref{dLCEL} in Proposition \ref{dLCEL-prop}. We have also determined the right and left Legendre transformations $\mathbb{F}^{\pm}L^d_\alpha$ in \eqref{legtrans-loc}. In this subsection, our goal is to derive discrete versions of the right and the left discrete Hamiltonian dynamics exhibited in Proposition \ref{dHam-main}.

Let us start with the local picture first by concentrating on the plus momenta. Accordingly, we pass to the Hamiltonian picture applying the right discrete Legendre transformation to the discrete locally conformally Euler-Lagrange equations \eqref{dLCEL}. At first, we introduce the Hamiltonian function for $\mathbf{r}$-momenta as 
\begin{equation}\label{pelin}
    {H_\alpha^d}^+(\mathbf{q}_\kappa, \mathbf{r}_{\kappa+1})=  \mathbf{r}_{\kappa+1} \cdot  \mathbf{q}_{\kappa+1}- L^d_{\alpha}(\mathbf{q}_\kappa, \mathbf{q}_{\kappa + 1}).
\end{equation}
A direct computation reads the local right discrete Hamiltonian dynamics in terms of the $\mathbf{r}$-momenta  as  
\begin{equation}\label{rdhe}
\begin{cases}
& \mathbf{q}_{\kappa+1}=D_2{H_\alpha^d}^+(\mathbf{q}_\kappa,{\mathbf{r}}_{\kappa+1}),\nonumber\\
&{\mathbf{r}}_\kappa=D_1{H_\alpha^d}^+(\mathbf{q}_\kappa,{\mathbf{r}}_{\kappa+1}).
\end{cases}
\end{equation}
We assume the locally conformal character for the local discrete Hamiltonian functions. Motivated by the continuous case depicted in \eqref{Fatih}, we take
\begin{equation}\label{Fatih-d-local}
e^{\sigma_\alpha}H^{d+}_\alpha=e^{\sigma_\beta}H^{d+}_\beta
\end{equation}
on the intersection of the tangent bundles $TV_\alpha \cap TV_\beta$. The equality determines an obstruction to glue up local functions $\{H^d_\alpha\}$ to a real-valued function on the cotangent bundle $T^*Q$. However, the same identity determines a global discrete Hamiltonian function $H^d$ whose local characterization is given by in terms of the $\mathbf{p}$-momenta as  
\begin{equation}\label{Fatih-discrete-r}
H^{d+}\vert_\alpha(\mathbf{q}_\kappa, \mathbf{p}_{\kappa+1}):=  e^{\sigma_{\alpha}(\mathbf{q}_\kappa)} {H_\alpha^d}^+(\mathbf{q}_\kappa, \mathbf{r}_{\kappa+1}).
\end{equation}
Here the $\mathbf{r}$-momenta and the $\mathbf{p}$-momenta are related as shown in  \eqref{r-p-dd}. 

\begin{proposition}
    The right discrete locally conformal Hamiltonian dynamics generated by the right discrete Hamiltonian function $H^{d+}$ given in \eqref{Fatih-discrete-r} is 
    \begin{equation}\label{r-dLCH} 
\begin{cases}
&\mathbf{q}_{\kappa + 1}= e^{\sigma_{\alpha}(\mathbf{q}_{\kappa+1})-\sigma_{\alpha}(\mathbf{q}_{\kappa})} D_2 H^{d+}\vert_{\alpha}(\mathbf{q}_{\kappa},\mathbf{p}_{\kappa+1}), \\
& \mathbf{p}_{\kappa}= D_1 H^{d+}\vert_{\alpha}(\mathbf{q}_{\kappa},\mathbf{p}_{\kappa+1}) -D\sigma_{\alpha}(\mathbf{q}_{\kappa})\cdot  H^{d+}\vert_{\alpha}(\mathbf{q}_{\kappa},\mathbf{p}_{\kappa+1}) .
\end{cases}
\end{equation}
\end{proposition}
\begin{proof}
Referring to the global discrete Lagrangian function $L^d$ in \eqref{dgLagfunc}, the equation \eqref{pelin} can be written as 
\begin{equation} \label{assaa}
\begin{split}
  &     e^{-\sigma_{\alpha}(\mathbf{q}_{\kappa})}H^{d+}\vert_\alpha(\mathbf{q}_\kappa, \mathbf{p}_{\kappa+1})\\&  \qquad \qquad  =  e^{-\sigma_{\alpha}(\mathbf{q}_{\kappa+1})}\mathbf{p}_{\kappa+1} \cdot \mathbf{q}_{\kappa+1}- e^{-\sigma_{\alpha}(\mathbf{q}_{\kappa})}L^d|_{\alpha}(\mathbf{q}_\kappa, \mathbf{q}_{\kappa + 1}).
\end{split}
  \end{equation}
Taking the partial derivative of this equation with respect to its first entry $\mathbf{q}_{\kappa}$, we arrive at 
\begin{equation}
    \begin{split}
        & - e^{-\sigma_{\alpha}(\mathbf{q}_{\kappa})}D\sigma_\alpha(\mathbf{q}_\kappa)H^{d+}\vert_\alpha(\mathbf{q}_\kappa, \mathbf{p}_{\kappa+1})+ e^{-\sigma_{\alpha}(\mathbf{q}_{\kappa})}D_1H^{d+}\vert_\alpha(\mathbf{q}_\kappa, \mathbf{p}_{\kappa+1}))\\
        & \qquad \qquad =e^{-\sigma_{\alpha}(\mathbf{q}_{\kappa})}(D\sigma_\alpha(\mathbf{q}_\kappa)L^d|_{\alpha}(\mathbf{q}_\kappa, \mathbf{q}_{\kappa + 1})-D_1L^d|_{\alpha}(\mathbf{q}_\kappa, \mathbf{q}_{\kappa + 1}))
        \end{split}
        \end{equation}
 then use the definition $\mathbf{p}_\kappa$ in \eqref{p-momenta} we get the second line in \eqref{r-dLCH}. In a similar manner, to obtain the first line in \eqref{r-dLCH}, we differentiate \eqref{assaa} with respect to the second entry $\mathbf{p}_{\kappa+1}$, and a direct calculation yields the result. 
\end{proof}
Notice that if one assumes that the conformal functions $\sigma_{\alpha}(\mathbf{q}_{\kappa})$ are constant functions, then the right discrete locally conformally symplectic Hamiltonian dynamics in \eqref{r-dLCH} reduce to the discrete Hamilton's equation \eqref{rd-Hameq}.

Equivalently, with the left Legendre transformation, we can obtain the local left discrete Hamiltonian
\begin{equation}\label{left-discham}
    {H_\alpha^d}^-(\mathbf{q}_{\kappa+1},\mathbf{r}_{\kappa})=-{\mathbf{r}}_\kappa \mathbf{q}_\kappa-L_\alpha^d(\mathbf{q}_\kappa,\mathbf{q}_{\kappa+1})
\end{equation}
and the local left discrete Hamilton's equations
\begin{equation}\label{ldhe}
\begin{cases}
 &\mathbf{q}_{\kappa}=-D_2{H_\alpha^d}^-(\mathbf{q}_{\kappa+1},{\mathbf{p}}_\kappa),\nonumber\\
&{\mathbf{r}}_{\kappa+1}=-D_1{H_\alpha^d}^-(\mathbf{q}_{\kappa+1},{\mathbf{p}}_\kappa).
 \end{cases}
\end{equation}
We express the relation \eqref{Fatih} for the left local discrete Hamiltonian function to define the left global discrete Hamiltonian function $H^{d-}$ with the following local expression:
\begin{equation}\label{Fatih-discrete-l}
H^{d-}\vert_\alpha(\mathbf{q}{\kappa+1}, \mathbf{p}{\kappa}):= e^{\sigma_{\alpha}(\mathbf{q}{\kappa+1})} {H\alpha^d}^-(\mathbf{q}{\kappa+1}, \mathbf{r}{\kappa}).
\end{equation}
This leads us to write \eqref{left-discham} in terms of the global discrete Hamiltonian function as 
\begin{equation}
     e^{-\sigma_{\alpha}(\mathbf{q}_{\kappa+1})}H^{d-}\vert_\alpha(\mathbf{q}_{\kappa+1}, \mathbf{p}_{\kappa})=  -e^{-\sigma_{\alpha}(\mathbf{q}_{\kappa})}\mathbf{p}_{\kappa} \cdot \mathbf{q}_{\kappa}- e^{-\sigma_{\alpha}(\mathbf{q}_{\kappa})}L^d|_{\alpha}(\mathbf{q}_\kappa, \mathbf{q}_{\kappa + 1}).
\end{equation}
By performing similar calculations done for the case of the right Hamiltonian, we obtain the left discrete LCS Hamilton's equation stated in the following proposition.
\begin{proposition}
    The left discrete LCS Hamilton's equation generated by the left discrete Hamiltonian function $H^{d-}$ in \eqref{Fatih-discrete-l} is 
    \begin{equation}\label{l-dLCH}
\begin{cases}
&\mathbf{q}_{\kappa}=-e^{\sigma{\alpha}(\mathbf{q}_{\kappa})-\sigma{\alpha}(\mathbf{q}_{\kappa+1})} D_2 H^{d-}\vert_{\alpha}(\mathbf{q}_{\kappa+1},\mathbf{p}_{\kappa}), \\
& \mathbf{p}_{\kappa+1}=e^{\sigma{\alpha}(\mathbf{q}_{\kappa})-\sigma{\alpha}(\mathbf{q}_{\kappa+1})}(D\sigma{\alpha}(\mathbf{q}_{\kappa+1})\cdot H^{d-}\vert{\alpha}(\mathbf{q}_{\kappa+1},\mathbf{p}_{\kappa}) \\
&\qquad \qquad \qquad - D_1 H^{d-}\vert_{\alpha}(\mathbf{q}_{\kappa+1},\mathbf{p}_{\kappa})).
\end{cases}
\end{equation}
\end{proposition}
The following diagram gives the hierarchy of discrete and continuous Hamiltonian dynamics discussed in this work: 
\begin{equation}\label{d-h}
\xymatrix{{\begin{array}{c} LCS~ HE\\ \eqref{LCSHE} \end{array}}   
\ar[d]_{discretization} \ar[rrr]^{\varphi=0}&   & &  {\begin{array}{c} ~HE\\
\eqref{HE}  \end{array}} \ar[d]_{discretization} \\
{\begin{array}{c} dLCS ~ HE \\ right: \eqref{r-dLCH} \\  left:\eqref{l-dLCH}  \end{array}}  \ar[rrr]^{\varphi=0}&  & & {\begin{array}{c} dHE \\ right: \eqref{rd-Hameq} \\  left:\eqref{ld-Hameq} \end{array}} }
\end{equation}
where $d$ refers to term discrete whereas $HE$ stands for an abbreviation of Hamilton's equation.

\section{Conclusion and Future Work}
In this work, we have examined the globalization problem of discrete Lagrangian and Hamiltonian dynamics in the locally conformal framework. We first presented the locally conformal Euler-Lagrange dynamics \eqref{LCEL} and discretized it in \eqref{dLCEL}. Diagram \ref{d-d} demonstrates this at a glance. Additionally, we provided the discretization of LCS Hamiltonian dynamics \eqref{r-dLCH} and \eqref{l-dLCH}, illustrated in Diagram \ref{d-h}.

As a future work, we plan to generalize the present approach to include time-dependent Lagrangian dynamics. Although discretization techniques for such Lagrangian and Hamiltonian systems are available in the literature, see, for example, \cite{Colombo-discrete-cosymplectic,deLeon-time-discrete}, the globalization problem remains unexplored. Cosymplectic manifolds are well-suited for studying time-dependent dynamics \cite{Cape}.  To address the globalization problem in this context, locally conformally cosymplectic manifolds, as described in \cite{AtEsLeSa23,chinea91}, could provide a suitable framework. Another future work to is write geometric Hamilton-Jacobi theory for the discrete locally conformal dynamics. For the continuous locally conformal dynamics the Hamilton-Jacobi problem is studied in a series of papers \cite{esen2022reviewing,EsLeSaZa-Cauchy,EsLeSaZa21b,EsenLeonSarZaj1} ranging from the particle motion to the fields theories.

\bibliographystyle{abbrv}
  
\bibliography{references}

\end{document}